\documentclass[10pt, conference]{IEEEtran}
\IEEEoverridecommandlockouts
\usepackage{kotex}
\usepackage{ifpdf}
\usepackage{cite}
\usepackage{bbm}
\ifCLASSINFOpdf
\usepackage{graphicx}
\usepackage{subfigure}
\usepackage{epstopdf}
\usepackage{xcolor}
\else
\fi
\usepackage{float}
\usepackage{amsmath}
\usepackage{amssymb}
\usepackage{amsthm}
\usepackage{mathtools}

\newtheorem{prop}{Proposition}
\newtheorem{thm}{Theorem}
\newtheorem{remark}{Remark}

\usepackage{algorithm}
\usepackage{algpseudocode}

\allowdisplaybreaks
\usepackage{array}
\usepackage{stfloats}
\usepackage{tikz}
\theoremstyle{definition}
\newtheorem{definition}{Definition}
\usepackage[shortlabels]{enumitem}
\usepackage{comment}
\usepackage{multirow}
\usepackage{booktabs}

\usepackage{enumitem}

\def\BibTeX{{\rm B\kern-.05em{\sc i\kern-.025em b}\kern-.08em
    T\kern-.1667em\lower.7ex\hbox{E}\kern-.125emX}}
\begin{document}
\title{Differential Privacy in Feature Reconstruction Aided Federated Learning for Agent’s Semantic Communication Model Update\\
\thanks{This work was supported by Institute of Information \& communications Technology Planning \& Evaluation (IITP) grant funded by the Korea government(MSIT). (No.RS-2024-00398948, Next Generation Semantic Communication Network Research Center)}
}

\author{\IEEEauthorblockN{Yoon Huh\IEEEauthorrefmark{1}\IEEEauthorrefmark{2}, Bumjun Kim\IEEEauthorrefmark{1}\IEEEauthorrefmark{2}, and Wan Choi\IEEEauthorrefmark{1}\IEEEauthorrefmark{2}}
\IEEEauthorblockA{\IEEEauthorrefmark{1}Department of Electrical and Computer Engineering, Seoul National University, Seoul 08826, South Korea}
\IEEEauthorblockA{\IEEEauthorrefmark{2}Institute of New Media and Communications, Seoul National University, Seoul 08826, South Korea}
\{mnihy621, eithank96, wanchoi\}@snu.ac.kr
}

\maketitle

\begin{abstract}
This paper proposes a differentially private federated learning (FL) framework built upon an FL algorithm with semantic feature reconstruction (FedSFR) for training semantic communication modules for image transmission. By allowing clients with unfavorable uplink capacity to transmit low-dimensional semantic feature vectors extracted from locally trained joint source-channel coding (JSCC) encoders, FedSFR enhances communication efficiency and training stability under heterogeneous wireless conditions. To protect client privacy, we incorporate the oneshot Laplace mechanism and theoretically demonstrate that feature-based transmission achieves strictly stronger differential privacy (DP) guarantees than gradient-based transmission under an identical communication budget. In addition, a model selection mechanism is introduced to alleviate performance degradation caused by privacy-preserving perturbations. Experimental results on multiple datasets show that the proposed DP-aided FedSFR outperforms DP-enabled FedAvg in training stability and image reconstruction quality in heterogeneous wireless systems.
\end{abstract}


\vspace{-5mm}
\section{Introduction}\label{sec:introduction}
Semantic communication has emerged as a key paradigm for task-oriented transmission in sixth-generation (6G) wireless systems, aiming to convey task-relevant information rather than raw bit streams \cite{luo2022semantic}. For robust image transmission, deep learning (DL)–based joint source–channel coding (JSCC) frameworks have been widely studied, where neural network (NN)–based autoencoders jointly optimize source and channel coding over noisy wireless channels \cite{huh2025universal}.

To address the dynamic nature of real-world data distributions while preserving data privacy, federated learning (FL) has emerged as an effective solution for continuously updating semantic communication modules across distributed agents \cite{sun2024federated, xu2024federated}. In our previous work, we proposed an FL algorithm with semantic feature reconstruction (FedSFR), a communication-efficient FL framework in which clients experiencing poor channel conditions transmit semantic feature vectors instead of local gradient updates \cite{huh2025feature, huh2026federated}. By combining gradient aggregation from clients with better channel conditions with server-side feature reconstruction (FR) learning, FedSFR significantly \textit{accelerates convergence} and \textit{enhances training stability} in heterogeneous wireless environments. Within this framework, semantic communication agents operate as FL clients and update their JSCC encoder–decoder pair, i.e., the local model, through FL. During training, the parameter server (PS) aggregates local updates and refines the global model via FR. After training, however, inter-agent semantic communication is conducted exclusively among clients using the learned global model, without involvement of the PS.

Although FL avoids direct sharing of raw data, prior studies have shown that privacy leakage may still occur through model updates or learned representations \cite{kim2025privacy}. In FedSFR, feature vectors are generated from public data, so the input samples themselves do not pose a privacy risk. However, since the encoder used for feature generation is trained on private local data, the resulting feature representations may still implicitly reflect client-specific information.

To address this concern, this paper investigates the differential privacy (DP) \cite{dwork2014algorithmic} guarantees of FedSFR under an honest-but-curious (HBC) PS. DP is a rigorous and widely adopted framework for protecting sensitive data while preserving the utility of privatized information and an HBC PS is assumed to faithfully follow the FedSFR protocol while attempting to infer clients’ original data from uplink communications. Relatedly, the DP method in \cite{liu2024adaptive} employs PS-injected noise to defend against an external adversary rather than an HBC PS, and it neither targets communication-efficient FL nor exploits feature vectors in semantic communication. In contrast, FedSFR supports not only gradient-based updates, as in conventional FL, but also feature-based transmission. Notably, in the feature-based transmission, feature generation relies solely on the encoder of an autoencoder-based semantic communication module. As a result, only the JSCC encoder is subject to privatization, whereas the gradient-based update of FedSFR requires privatizing the entire model. Since DP guarantees are governed by the sensitivity of the privatized target, this structural distinction motivates a fundamental question: \textit{how does the choice of transmission strategy affect the achievable DP guarantees in FedSFR?}

To answer this question, we integrate the oneshot and Laplace mechanisms, dubbed the oneshot Laplace (OL) mechanism \cite{qiao2021oneshot}, for both transmission strategies and provide a rigorous theoretical comparison of their privacy guarantees. We mathematically prove that feature-based transmission achieves strictly stronger DP guarantees under identical communication constraints due to reduced parameter exposure. Furthermore, to mitigate performance degradation caused by DP noise, we propose a model selection (MS) algorithm that compares FR-assisted updates with gradient-only updates, and then selectively retains the update achieving better task performance.

The main contributions of this work are as follows:
\begin{itemize}
    \item We propose a differentially private FedSFR by incorporating the OL mechanism into both gradient-based and feature-based transmission strategies, enabling privacy-preserving training under heterogeneous uplink channels.

    \item We propose the MS algorithm at the PS that adaptively selects between FR-aided and gradient-only updates, improving training stability and yielding a more favorable optimization trajectory under DP noise.
    
    \item We establish formal DP guarantees for FedSFR under both sparse gradient and semantic feature transmission. We show that feature-based transmission allows strictly stronger privacy guarantees by reducing the effective parameter exposure under the same uplink budget.

    \item Experimental results on two image datasets demonstrate that the proposed DP-aided FedSFR consistently outperforms DP-enabled FedAvg in terms of reconstruction quality, convergence behavior, and robustness under various privacy and channel conditions.
\end{itemize}

\vspace{-2mm}
\section{System Model}\label{sec:system model}

\subsection{Semantic Communication for Image Transmission}\label{subsec:semantic communication for image transmission}
Consider a JSCC-based semantic communication framework for image transmission, where the receiver aims to recover the transmitted image. The transmission procedure is summarized as follows. Given a source image $\mathbf{X} \in \mathbb{R}^{C \times H \times W}$, with $C$, $H$, and $W$ denoting the number of feature channels, height, and width, respectively, the transmitter maps $\mathbf{X}$ into a $d$-dimensional feature vector $\mathbf{y} \in \mathbb{R}^d$ through an encoder $f_{\boldsymbol{\theta}}$ parameterized by $\boldsymbol{\theta}$, expressed as $\mathbf{y} = f_{\boldsymbol{\theta}}(\mathbf{X})$.

To satisfy the transmission power constraint, assumed as $P = 1$, the encoded feature vector is normalized as $\tilde{\mathbf{y}} = \mathbf{y}/\|\mathbf{y}\|_2$ and conveyed over a Rayleigh fading channel with coefficient $h \sim \mathcal{CN}(0, 1)$. At the receiver, after channel equalization, the corrupted representation is processed by a decoder $f^{-1}_{\boldsymbol{\phi}}$, parameterized by $\boldsymbol{\phi}$, to generate the reconstructed image $\hat{\mathbf{X}} \in \mathbb{R}^{C \times H \times W}$, given by $\hat{\mathbf{X}} = f^{-1}_{\boldsymbol{\phi}}(\tilde{\mathbf{y}} + \mathbf{n}/h)$, where $\mathbf{n} \sim \mathcal{N}(\boldsymbol{0}_d, \sigma^2 \mathbf{I}_d)$ denotes the Gaussian noise vector. The average signal-to-noise ratio (SNR) $\gamma$ is defined as $1/\sigma^2$.

Let $\boldsymbol{w} = \{\boldsymbol{\theta}, \boldsymbol{\phi}\} \in \mathbb{R}^N$ collect all trainable parameters of the JSCC model, where $N$ represents the total number of NN parameters. In the FL setting, these parameters are optimized by minimizing the global objective function at the PS, defined as $F(\boldsymbol{w}) = \frac{1}{|\mathcal{D}|} \sum_{\mathbf{X} \in \mathcal{D}} l_c(\boldsymbol{w}; \mathbf{X})$, where $\mathcal{D}$ denotes the image dataset. Here, the sample-wise loss function is chosen as $l_c(\boldsymbol{w}; \mathbf{X}) = \mathsf{MSE}(\hat{\mathbf{X}}, \mathbf{X})$, which measures the mean squared error (MSE) between the reconstructed image $\hat{\mathbf{X}}$ and the ground-truth image $\mathbf{X}$.

\vspace{-2mm}
\subsection{Federated Learning with Semantic Feature Reconstruction}\label{subsec:federated learning with semantic feature reconstruction}

\begin{figure}[!t]
    \centering
    \includegraphics[width=0.9\linewidth]{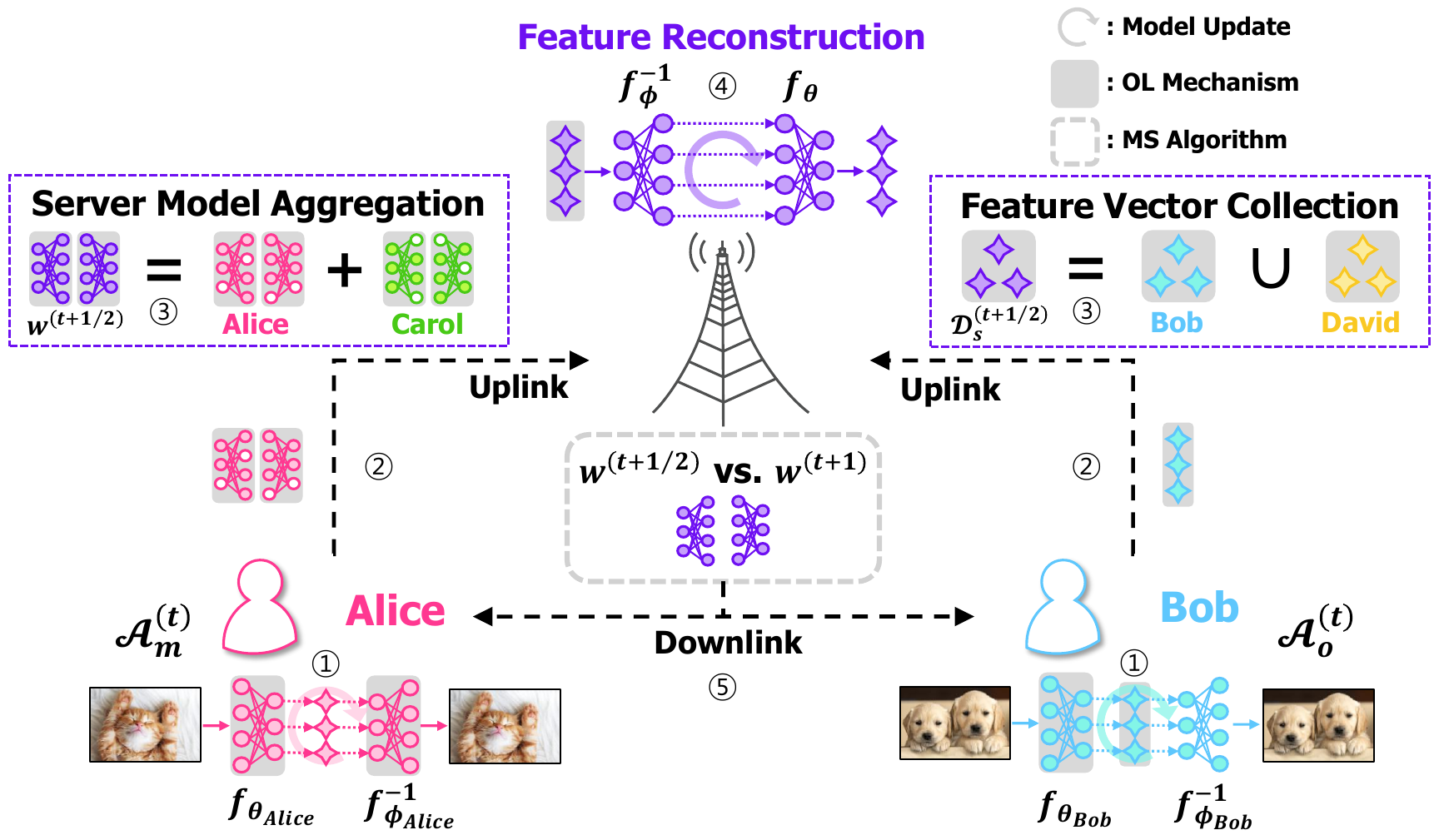}
    \vspace{-3mm}
    \caption{Overall procedure of FedSFR with the numbered algorithmic steps. Gray boxes indicate the proposed OL mechanism and MS algorithm for DP.}
    \label{fig:Overall procedure of FedSFR with the numbered algorithmic steps}
    \vspace{-6mm}
\end{figure}

We consider a wireless FedSFR system \cite{huh2025feature}, which departs from conventional FL schemes that rely solely on gradient-based uplink transmission. The PS interacts with a set of $K$ clients indexed by $\mathcal{A}$ with $|\mathcal{A}| = K$. Each client $k\in\mathcal{A}$ maintains a local dataset $\mathcal{D}_k$, while $\mathcal{D} = \bigcup_{k\in\mathcal{A}}\mathcal{D}_k$ is the global dataset. The local objective function at client $k$ is defined as $F_k(\boldsymbol{w}) = \frac{1}{|\mathcal{D}_k|}\sum_{\mathbf{X}\in\mathcal{D}_k} l_c(\boldsymbol{w}; \mathbf{X})$, whereas the global objective function at the PS is expressed as $F(\boldsymbol{w}) = \sum_{k\in\mathcal{A}}p_k F_k(\boldsymbol{w})$, where $p_k = |\mathcal{D}_k| / |\mathcal{D}|$ denotes the relative data contribution of client $k$. The overall FedSFR workflow is illustrated in Fig. \ref{fig:Overall procedure of FedSFR with the numbered algorithmic steps}, including five algorithmic \texttt{steps}.

In FedSFR, clients with unfavorable uplink capacity upload encoder output feature vectors extracted from their locally updated JSCC encoders instead of transmitting local model (i.e., JSCC module) updates. Since the feature dimension is much smaller than the number of model parameters, i.e., $d \ll N$, feature-based transmission enables the clients with limited uplink capacity to convey more informative representations of their local model states than sparsified local updates under the same communication budget. At each global iteration $t$, the participating clients $\mathcal{A}^{(t)} \subset \mathcal{A}$ are partitioned into two disjoint sets, $\mathcal{A}_m^{(t)}$ and $\mathcal{A}_o^{(t)}$, where clients in $\mathcal{A}_m^{(t)}$, experiencing relatively reliable channel conditions, transmit local \underline{\textit{m}}odel updates, whereas the clients in $\mathcal{A}_o^{(t)}$, suffering from poorer channel quality, upload encoder \underline{\textit{o}}utput feature vectors. Accordingly, we have $\mathcal{A}_m^{(t)} \cup \mathcal{A}_o^{(t)} = \mathcal{A}^{(t)}$, with $|\mathcal{A}_m^{(t)}| = K_m$ and $|\mathcal{A}_o^{(t)}| = K_o$. Following \cite{zhao2018federated}, features are extracted from a shared public dataset $\mathcal{P}_k \subset \mathcal{D}_k$ to avoid direct leakage of local private data at the PS, while the local model is trained using the entire local dataset, including both private and public samples\footnote{If the channel conditions remain poor, some clients only transmit feature vectors generated from the public dataset, raising concerns about imbalanced influence of private data across clients, akin to heterogeneous local data distributions. However, unlike conventional FL for classification, our prior FL results on reconstruction under non-independent and identically distributed (IID) local datasets are similar to those under IID settings \cite{huh2026federated}, indicating that the bias caused by clients with persistently poor channels is negligible.}. Notably, the transmitted representations thus serve as a medium for conveying local model knowledge.

During the local update, each client $k$ downloads the current global model from the PS, denoted by $\boldsymbol{w}_k^{\scriptstyle(t, 0)} = \boldsymbol{w}^{\scriptstyle(t)} \in \mathbb{R}^N$. The client then performs local model optimization using its local dataset $\mathcal{D}_k$ via a mini-batch stochastic gradient descent (SGD) algorithm (\texttt{step 1}). The local update rule follows
\begin{align}\label{eq:local iteration}
    \boldsymbol{w}_k^{\scriptstyle(t, e + 1)} 
    = \boldsymbol{w}_k^{\scriptstyle(t, e)} 
    - \eta_c^{\scriptstyle(t)} 
    \nabla F_k^{\scriptstyle(t, e)}(\boldsymbol{w}_k^{\scriptstyle(t, e)}),
\end{align}
where $e \in \{0, \dots, E_c - 1\}$ denotes the local iteration index, $\eta_c^{\scriptstyle(t)}$ is the local learning rate at global round $t$, $E_c$ is the total number of local update steps, and $\nabla F_k^{\scriptstyle(t, e)}(\boldsymbol{w}_k^{(t, e)})$ is the stochastic gradient estimator. To compute this gradient estimator at each local iteration $e$, the client $k$ draws a mini-batch $\mathcal{D}_k^{(t, e)} \subset \mathcal{D}_k$ and averages the sample-wise gradients of $l_c(\boldsymbol{w}_k^{(t, e)};\mathbf{X})$ over $\mathcal{D}_k^{(t, e)}$.

For clients in $\mathcal{A}_m^{\scriptstyle(t)}$, local model updates are communicated using a top-$S$ sparsification scheme\footnote{Due to the limited uplink capacity between the PS and each client, the amount of transmitted update information must be constrained to guarantee error-free communication within the required delay $T$. Accordingly, the sparsification level is determined as $S_k \!=\! \lfloor\frac{C_k\cdot T\cdot R}{Q}\rfloor$, where $C_k$, $R$, and $Q$ are the channel capacity of client $k$, code rate, and quantization bits, respectively. $S$ may vary across clients depending on their channel conditions. The top-$S$ sparsification selects the $S$ entries with the largest magnitudes in a vector, i.e., $S_k$ for client $k$, and is applied on a per-layer basis for the NN model.} with error feedback strategy \cite{amiri2020federated}. Specifically, the local update information at global iteration $t$ is given by
\begin{align}\label{eq:sparsification}
    \mathbf{g}_k^{\scriptstyle(t)} = \textsf{Sparse}(\mathbf{m}_k^{\scriptstyle(t)} + \eta_c^{\scriptstyle(t)}\sum\nolimits_{e = 0}^{E_c - 1}\nabla F_k^{\scriptstyle(t, e)}(\boldsymbol{w}_k^{\scriptstyle(t, e)})),
\end{align}
where $\textsf{Sparse}(\cdot)$ denotes the top-$S$ sparsification operator and $\mathbf{m}_k^{\scriptstyle(t)} \in \mathbb{R}^N$ is the error memory maintained at client $k$. The error memory is updated as
\begin{align}\label{eq:error memory update}
    \mathbf{m}_k^{\scriptstyle(t + 1)} = \mathbf{m}_k^{\scriptstyle(t)} + \eta_c^{\scriptstyle(t)}\sum\nolimits_{e = 0}^{E_c - 1}\nabla F_k^{\scriptstyle(t, e)}(\boldsymbol{w}_k^{\scriptstyle(t, e)}) - \mathbf{g}_k^{\scriptstyle(t)},
\end{align}
which compensates for information loss via \eqref{eq:sparsification}. With the sparsification level $S_k$ larger than the threshold $S_\mathsf{th}$, each client $k \in \mathcal{A}_m^{\scriptstyle(t)}$ uploads the local update $\mathbf{g}_k^{\scriptstyle(t)}$ to the PS via the uplink channel (\texttt{step 2}).

In contrast, the clients in $\mathcal{A}_o^{\scriptstyle(t)}$ upload a feature set $\mathcal{Y}_k^{\scriptstyle(t)}$ via the uplink, where each feature vector $\mathbf{y}$ is generated by the locally updated JSCC encoder $f_{\boldsymbol{\theta}}$ with $\boldsymbol{\theta} = \boldsymbol{\theta}_k^{\scriptstyle(t, E_c)}$ using $\mathcal{P}_k$ (\texttt{step 2}). Their uplink budget $S_k$\footnote{$S_k$ denotes the uplink payload size of client $k$. For $k\in\mathcal{A}_o^{\scriptstyle(t)}$, it determines the number of transmitted feature vectors, while for $k\in\mathcal{A}_m^{\scriptstyle(t)}$ it corresponds to the sparsification level for gradient transmission.} is smaller than $S_\mathsf{th}$, limiting the number of transmitted feature vectors.

During the global update, the PS first aggregates the received local updates and forms a feature-based dataset $\mathcal{D}_s^{\scriptstyle(t+\frac{1}{2})} = \bigcup_{k\in\mathcal{A}_o^{\scriptstyle(t)}} \mathcal{Y}_k^{\scriptstyle(t)}$ by collecting the transmitted features (\texttt{step 3}). The intermediate global model is computed as
\begin{align}\label{eq:global aggregation}
    \boldsymbol{w}^{\scriptstyle(t + \frac{1}{2})} = \boldsymbol{w}^{\scriptstyle(t)} - \frac{K}{K_m}\sum\nolimits_{k\in\mathcal{A}_m^{\scriptstyle(t)}}p_k \mathbf{g}_k^{\scriptstyle(t)}.
\end{align}
Starting from $\boldsymbol{w}_s^{\scriptstyle(t + \frac{1}{2}, 0)} = \boldsymbol{w}^{\scriptstyle(t + \frac{1}{2})}$, the PS further refines the global model by performing FR using $\mathcal{D}_s^{\scriptstyle(t + \frac{1}{2})}$ (\texttt{step 4}). Through FR, information from the clients in $\mathcal{A}_o^{\scriptstyle(t)}$ is implicitly integrated into the global model without explicitly transmitting their local updates. Accordingly, the error memory of the client $k \in \mathcal{A}_o^{\scriptstyle(t)}$ is reset as $\mathbf{m}_k^{\scriptstyle(t + 1)} = \boldsymbol{0}_N$ after the server-side update.

The JSCC decoder and encoder at the PS sequentially process a feature sample $\mathbf{y}$ in $\mathcal{D}_s^{\scriptstyle(t + \frac{1}{2})}$ to produce its reconstruction $\hat{\mathbf{y}}\in\mathbb{R}^{d}$, given by $\hat{\mathbf{y}} = f_{\boldsymbol{\theta}}(f^{-1}_{\boldsymbol{\phi}}(\Tilde{\mathbf{y}} + \mathbf{n}))$, where $\boldsymbol{\theta} = \boldsymbol{\theta}_s^{\scriptstyle(t + \frac{1}{2}, e)}$, $\boldsymbol{\phi} = \boldsymbol{\phi}_s^{\scriptstyle(t + \frac{1}{2}, e)}$, $\mathbf{n} \sim \mathcal{N}(\boldsymbol{0}_d, \sigma^2\mathbf{I}_d)$, and $\Tilde{\mathbf{y}} = \mathbf{y}/||\mathbf{y}||_2$. Using the same SGD method as adopted at the clients, the server update is expressed as
\begin{align}\label{eq:server iteration}
    \boldsymbol{w}_s^{\scriptstyle(t + \frac{1}{2}, e + 1)} = \boldsymbol{w}_s^{\scriptstyle(t + \frac{1}{2}, e)} - \eta_s^{\scriptstyle(t)}\nabla F_s^{\scriptstyle(t + \frac{1}{2}, e)}(\boldsymbol{w}_s^{\scriptstyle(t + \frac{1}{2}, e)}),
\end{align}
where $e\in\{0, \dots, E_s - 1\}$ denotes the server iteration index, $E_s$ is the total number of server updates, $\eta_s^{\scriptstyle(t)}$ is the server learning rate at global iteration $t$, and $\nabla F_s^{\scriptstyle(t + \frac{1}{2}, e)}(\boldsymbol{w}_s^{\scriptstyle(t + \frac{1}{2}, e)})$ is the stochastic gradient estimator. To obtain this gradient estimator at each server iteration $e$, the PS samples a mini-batch $\mathcal{D}_s^{\scriptstyle(t + \frac{1}{2}, e)}\subset\mathcal{D}_s^{\scriptstyle(t + \frac{1}{2})}$ and averages the sample-wise gradients of the reconstruction loss $l_s(\boldsymbol{w};\mathbf{y}) = \mathsf{MSE}(\hat{\mathbf{y}},\mathbf{y})$ over $\mathcal{D}_s^{(t+\frac{1}{2}, e)}$. Finally, the PS broadcasts the global model $\boldsymbol{w}^{\scriptstyle(t + 1)} = \boldsymbol{w}_s^{\scriptstyle(t + \frac{1}{2}, E_s)}$ to all clients (\texttt{step 5}).

\vspace{-2mm}
\section{Differentially Private FedSFR}\label{sec:differentially private fedsfr}

\subsection{Oneshot Laplace Mechanism}\label{subsec:oneshot laplace mechanism}
Among various variants of DP, the oneshot and Laplace mechanisms are particularly well suited for top-$S$ selection and wireless communication settings due to their computational efficiency and simplicity \cite{dwork2014algorithmic, qiao2021oneshot}.  Specifically, the oneshot mechanism first privately selects the top-$S$ indices without revealing their true ordering, and the Laplace mechanism then perturbs the selected values to conceal their exact magnitudes.

Let $\mathbb{D}$ denote the data universe, and let $\mathbb{D}^m$ be the set of datasets consisting of $m\in\mathbb{Z}$ records. We say that two datasets $\mathcal{D}, \mathcal{D}' \in \mathbb{D}^m$ are adjacent if they differ in exactly one record. For a fixed positive integer $n$, consider a query function $\mathcal{Q}: \mathbb{D}^m\to\mathbb{R}^{n}$. The sensitivity of this query is defined as the maximum change in its output, when evaluated on any pair of adjacent datasets, given by $\Delta = \max\|\mathcal{Q}(\mathcal{D}) - \mathcal{Q}(\mathcal{D}')\|_1$.

\begin{definition}[Differential privacy \cite{dwork2014algorithmic}]
For a privacy parameter $\epsilon > 0$, a randomized mechanism $\mathcal{M}:\mathbb{R}^n \to \mathbb{R}^n$ satisfies $\epsilon$-DP if, for all adjacent databases $\mathcal{D}, \mathcal{D}'$ and every measurable set $J\subseteq \mathbb{R}^n$, $\Pr[\mathcal{M}(\mathcal{Q}(\mathcal{D}))\in J] \leq e^\epsilon \Pr[\mathcal{M}(\mathcal{Q}(\mathcal{D}'))\in J]$ holds.
\end{definition}

\begin{prop}[Oneshot mechanism \cite{qiao2021oneshot}]\label{pro:oneshot mechanism}
\emph{Oneshot mechanism $\mathcal{M}_{os}$ adds Laplace noise $\mathcal{L}_1 \sim \mathrm{Lap}(0, \sigma_1)^n$ to a query $\mathcal{Q}$ and returns the top-$S$ indices, achieving $\epsilon$-DP with $\epsilon = \frac{2S\Delta}{n\sigma_1}$.}
\end{prop}

\begin{prop}[Laplace mechanism \cite{dwork2014algorithmic}]\label{pro:laplace mechanism}
\emph{For a query $\mathcal{Q}$, the Laplace mechanism is defined as $\mathcal{M}_L(\mathcal{Q}(\mathcal{D})) \triangleq \mathcal{Q}(\mathcal{D}) + \mathcal{L}_2$, where $\mathcal{L}_2 \sim \mathrm{Lap}(0, \sigma_2)^n$, achieving $\epsilon$-DP with $\epsilon = \frac{\Delta}{\sigma_2}$.}
\end{prop}

Based on these foundations, we describe two distinct transmission strategies under an identical uplink communication budget to enable privacy-preserving uplink transmission: (i) the clients in $\mathcal{A}_m^{\scriptstyle(t)}$ transmit the local update information $\mathbf{g}_k^{\scriptstyle(t)}$ and (ii) the clients in $\mathcal{A}_m^{\scriptstyle(t)}$ transmit the feature vectors $\mathcal{Y}_k^{\scriptstyle(t)}$.

\begin{enumerate}[label=(\roman*)]
    \item \textbf{Gradient-based Transmission with DP}: Each client $k$ transmits both the sparsified update values $\{[\mathbf{g}_k^{\scriptstyle(t)}]_{i_s}\}_{s = 1}^{S_k}$ and the corresponding indices of the top-$S_k$ selected elements $\{i_s\}_{s = 1}^{S_k}$, as in \cite{qiao2021oneshot}, where $i_s\in[1{:}N]$ and $[\mathbf{g}_k^{\scriptstyle(t)}]_{i_s}$ denote the $i_s$-th entry of $\mathbf{g}_{k}^{\scriptstyle(t)}$. To achieve DP, a two-stage mechanism is employed. First, the oneshot mechanism $\mathcal{M}_{\text{os}}$ is applied to privately select the top-$S_k$ indices from the Laplace-perturbed local updates $\mathbf{g}_k^{\scriptstyle(t)} + \boldsymbol{\ell}_{k, 1}^{\scriptstyle(t)}$, where $\boldsymbol{\ell}_{k, 1}^{\scriptstyle(t)} \sim \mathrm{Lap}(0, \sigma_1)^N$. Then, to obscure the magnitudes of the selected updates, fresh independent Laplace noise is added via the Laplace mechanism $\mathcal{M}_L$, yielding $[\mathbf{g}_k^{\scriptstyle(t)}]_{i_s} + [\boldsymbol{\ell}_{k, 2}^{\scriptstyle(t)}]_s$ for $s\in[1{:}S_k]$, where $[\boldsymbol{\ell}_{k, 2}^{\scriptstyle(t)}]_s$ denotes the $s$-th entry of $\boldsymbol{\ell}_{k, 2}^{\scriptstyle(t)} \sim \mathrm{Lap}(0, \sigma_2)^{S_k}$.

    \item \textbf{Feature-based Transmission with DP}: Unlike case (i), $\mathcal{M}_{\text{os}}$ and $\mathcal{M}_L$ are applied only to the encoder-side local updates. Hence, the encoder parameters are privatized prior to feature generation. Since the encoder contains about half of the parameters of the full JSCC model, we assume that it consists of $N/2$ parameters. The client $k$ privately selects the top-$(S_k/2)$ indices $\{i_s\}_{s = 1}^{S_k/2}$ from the Laplace-perturbed encoder updates $\mathbf{g}_{k, \boldsymbol{\theta}}^{\scriptstyle(t)} + \boldsymbol{\ell}_{k, 1}^{\scriptstyle(t)}$, and then corrupts their magnitudes as $[\mathbf{g}_{k, \boldsymbol{\theta}}^{\scriptstyle(t)}]_{i_s} + [\boldsymbol{\ell}_{k, 2}^{\scriptstyle(t)}]_s$ for $s\in[1{:}S_k/2]$. Here, $\mathbf{g}_{k, \boldsymbol{\theta}}^{\scriptstyle(t)}\in\mathbb{R}^{N/2}$ denotes the encoder-side local updates, with $[\mathbf{g}_{k, \boldsymbol{\theta}}^{\scriptstyle(t)}]_{i_s}$ representing its $i_s$-th entry, $\boldsymbol{\ell}_{k, 1}^{\scriptstyle(t)} \sim \mathrm{Lap}(0, \sigma_1)^{N/2}$ and $\boldsymbol{\ell}_{k, 2}^{\scriptstyle(t)} \sim \mathrm{Lap}(0, \sigma_2)^{S_k/2}$.
\end{enumerate}

We collectively refer to the privacy mechanisms used in cases (i) and (ii) as the OL mechanism. As illustrated in Fig. \ref{fig:Overall procedure of FedSFR with the numbered algorithmic steps}, the key distinction lies in the privatization target: (i) OL is applied to the full JSCC model; (ii) it is imposed only on the JSCC encoder parameters prior to feature generation.

\vspace{-1mm}
\subsection{Model Selection Algorithm}\label{subsec:model selection algorithm}
The FR process in FedSFR enables the global JSCC model to incorporate client-specific characteristics by leveraging feature vectors that implicitly encode local update information rather than explicit gradients. However, under the OL mechanism, the benefits of FR, including fast convergence and stable training, may be reduced. This limitation arises because the transmitted feature vectors are generated by a locally perturbed JSCC encoder, whose parameters have already been corrupted by injected noise. Consequently, reconstructing these noisy feature representations can drive the PS toward a suboptimal global model compared to learning from clean features produced by an unperturbed encoder.

To mitigate this issue, we propose a novel MS algorithm\footnote{Although this issue can be alleviated by proper control of $\sigma_1$ and $\sigma_2$, the resulting performance variations are difficult to characterize explicitly. Moreover, the MS algorithm is compatible with such noise-control strategies.} that compares the image reconstruction performance of the intermediate global model $\boldsymbol{w}^{\scriptstyle(t + \frac{1}{2})}$ and the FR-aided global model $\boldsymbol{w}^{\scriptstyle(t + 1)}$, as shown in Fig. \ref{fig:Overall procedure of FedSFR with the numbered algorithmic steps}. Since the PS has access to a shared public dataset $\mathcal{P} = \bigcup_{k\in\mathcal{A}}\mathcal{P}_k$, it can evaluate image reconstruction quality without accessing any private client data. Specifically, it measures the MSE of both models on the public dataset $\mathcal{P}$ and selects the one with lower error as the updated global model. This alleviates the negative impact of reconstructing the noisy features under the OL mechanism. By retaining the model with superior reconstruction at each iteration, the proposed MS encourages that the global training trajectory steadily improves image reconstruction, enhancing both training stability and robustness under privacy constraints.

\section{Differential Privacy Analysis}\label{sec:differential privacy analysis}
When clients transmit feature vectors extracted by a local JSCC encoder, fewer parameters are privatized compared to transmitting local model updates, which can offer inherent privacy advantages. To rigorously analyze the DP guarantees of the two private transmission strategies presented in Section \ref{subsec:oneshot laplace mechanism}, \textbf{Theorem \ref{thm:differential privacy}} is established by leveraging the following two properties together with the OL mechanism.

\begin{prop}[Composition property \cite{dwork2014algorithmic}]\label{prop:sequential mechanism}
\emph{Let $\mathcal{M}_1$ and $\mathcal{M}_2$ satisfy $\epsilon_1$-DP and $\epsilon_2$-DP, respectively. Then, the composed mechanism $\mathcal{M}_{2}\circ\mathcal{M}_{1}$ satisfies $(\epsilon_1 + \epsilon_2)$-DP.}
\end{prop}

\begin{prop}[Post-processing property \cite{dwork2014algorithmic}]\label{prop:post processing mechanism}
\emph{Let $\mathcal{M}$ satisfy $\epsilon$-DP. For any deterministic function $f$, the post-processed mechanism $f\circ\mathcal{M}$ also satisfies $\epsilon$-DP.}
\end{prop}

\begin{thm}\label{thm:differential privacy}
\emph{Assume that each client $k\in\mathcal{A}^{\scriptstyle(t)}$ adopts one of the two DP-enabled transmission strategies with the OL mechanism: (i) transmitting local update information $\mathbf{g}_k^{\scriptstyle(t)}\in\mathbb{R}^N$, or (ii) transmitting the feature vectors $\mathcal{Y}_k^{\scriptstyle(t)}$. Under the oneshot and Laplace mechanism, utilizing $\mathrm{Lap}(0, \sigma_1)$ and $\mathrm{Lap}(0, \sigma_2)$, respectively, the option (ii) achieves a strictly stronger DP guarantee than the option (i), under an identical uplink communication budget $S$.}
\begin{IEEEproof}
\emph{For the global iteration $t$, assume that the local update information $\mathbf{g}_k^{\scriptstyle(t)}$, defined in \eqref{eq:sparsification}, has a bounded norm, i.e., $\|\mathbf{g}_k^{\scriptstyle(t)}\|_1\leq G^{\scriptstyle(t)}$ for all $k$. On the one hand, for the clients transmitting $\mathbf{g}_k^{\scriptstyle(t)}$, $\mathcal{M}_{os}$ is applied to select the top-$S$ indices and $\mathcal{M}_{L}$ is performed within the selected values. For all adjacent local update information $\mathbf{g}_k^{\scriptstyle(t)}$ and $\mathbf{g}_k'^{\scriptstyle(t)}$, the sensitivity of $\mathcal{M}_{os}$ can be bounded as $\Delta_{1}^{\scriptstyle(i)} = \|\mathbf{g}_k^{\scriptstyle(t)} - \mathbf{g}_k'^{\scriptstyle(t)}\|_1 \leq 2G^{\scriptstyle(t)}$. By \textbf{Proposition \ref{pro:oneshot mechanism}}, it satisfies $\epsilon_{1}^{\scriptstyle(i)}$-DP, where $\epsilon_{1}^{\scriptstyle(i)} = \frac{4SG^{\scriptstyle(t)}}{N\sigma_1}$. Additionally, the sensitivity corresponding to $\mathcal{M}_{L}$ is $\Delta_2^{\scriptstyle(i)} = \|\mathbf{g}_k^{\scriptstyle(t)} - \mathbf{g}_k'^{\scriptstyle(t)}\|_1 \leq \frac{2SG^{\scriptstyle(t)}}{N}$ since we only consider the selected values. Using \textbf{Proposition \ref{pro:laplace mechanism}}, it satisfies $\epsilon_{2}^{\scriptstyle(i)}$-DP, where $\epsilon_{2}^{\scriptstyle(i)} = \frac{2SG^{\scriptstyle(t)}}{N\sigma_2}$. Therefore, by \textbf{Proposition \ref{prop:sequential mechanism}}, every client transmitting $\mathbf{g}_k^{(t)}$ satisfies $\epsilon^{\scriptstyle(i)}$-DP for every global iteration, where $\epsilon^{\scriptstyle(i)} = \epsilon_{1}^{\scriptstyle(i)} + \epsilon_{2}^{\scriptstyle(i)} = \frac{4SG^{\scriptstyle(t)}}{N\sigma_1} + \frac{2SG^{\scriptstyle(t)}}{N\sigma_2} = \frac{2SG^{\scriptstyle(t)}(\sigma_1 + 2\sigma_2)}{N\sigma_1\sigma_2}$.}

\emph{On the other hand, for the clients transmitting $\mathcal{Y}_k^{\scriptstyle(t)}$, since clients only use the updated encoder to extract features, $\mathcal{M}_{os}$ and $\mathcal{M}_{L}$ are applied to the encoder-side local update information before feature extraction process, which is deterministic. According to the \textbf{Proposition \ref{prop:post processing mechanism}}, the feature extraction maintains the same privacy guarantee. The sensitivity on the encoder of $\mathcal{M}_{os}$ is given by $\Delta_{1}^{\scriptstyle(ii)} = \|\mathbf{g}_{k, \boldsymbol{\theta}}^{\scriptstyle(t)} - \mathbf{g}_{k, \boldsymbol{\theta}}'^{\scriptstyle(t)}\|_1 \leq G^{\scriptstyle(t)}$, where $\mathbf{g}_{k, \boldsymbol{\theta}}^{(t)}$ represents the encoder-side local update information. Here, assuming a symmetric encoder–decoder architecture, $\|\mathbf{g}_{k, \boldsymbol{\theta}}^{\scriptstyle(t)}\|_1 \leq \frac{G^{\scriptstyle(t)}}{2}$ holds. The sensitivity of $\mathcal{M}_{L}$ is $\Delta_{2}^{\scriptstyle(ii)} = \|\mathbf{g}_{k, \boldsymbol{\theta}}^{\scriptstyle(t)} - \mathbf{g}_{k, \boldsymbol{\theta}}'^{\scriptstyle(t)}\|_1\leq \frac{SG^{\scriptstyle(t)}}{N}$. Hence, similar to the derivation of $\epsilon^{\scriptstyle(i)}$, each client transmitting $\mathcal{Y}_k^{\scriptstyle(t)}$ satisfies $\epsilon^{\scriptstyle(ii)}$-DP for every global iteration, where $\epsilon^{\scriptstyle(ii)} = \frac{SG^{\scriptstyle(t)}(\sigma_1 + 2\sigma_2)}{N\sigma_1\sigma_2}$.}

\emph{Finally, since $\epsilon^{\scriptstyle(i)} > \epsilon^{\scriptstyle(ii)}$, this completes the proof.}
\end{IEEEproof}
\end{thm}

\begin{remark}
\emph{In the proof of \textbf{Theorem \ref{thm:differential privacy}}, the DP advantage of FedSFR stems from the fact that feature vectors are generated solely by the JSCC encoder, which contains roughly half of the parameters of the full autoencoder. While in case (i) the DP mechanisms are applied to the entire JSCC model, in case (ii) they are applied only to the encoder, effectively reducing the mechanism’s sensitivity by about 50\%. Since DP guarantees are tightly coupled to the size and sensitivity of the NN components being privatized, this structural reduction substantially improves privacy. Consequently, even with DP mechanisms beyond the OL mechanism, encoder-only privatization in FedSFR provides a general DP advantage.}
\end{remark}

\vspace{-2mm}
\section{Numerical Results}\label{sec:numerical results}

\begin{figure*}[!t]
    \centering
    \setlength{\subfigcapskip}{-2mm}
    \subfigure[$(\sigma_1, \sigma_2) = (0.0001, 0.0001)$]{\includegraphics[width=0.3\linewidth]{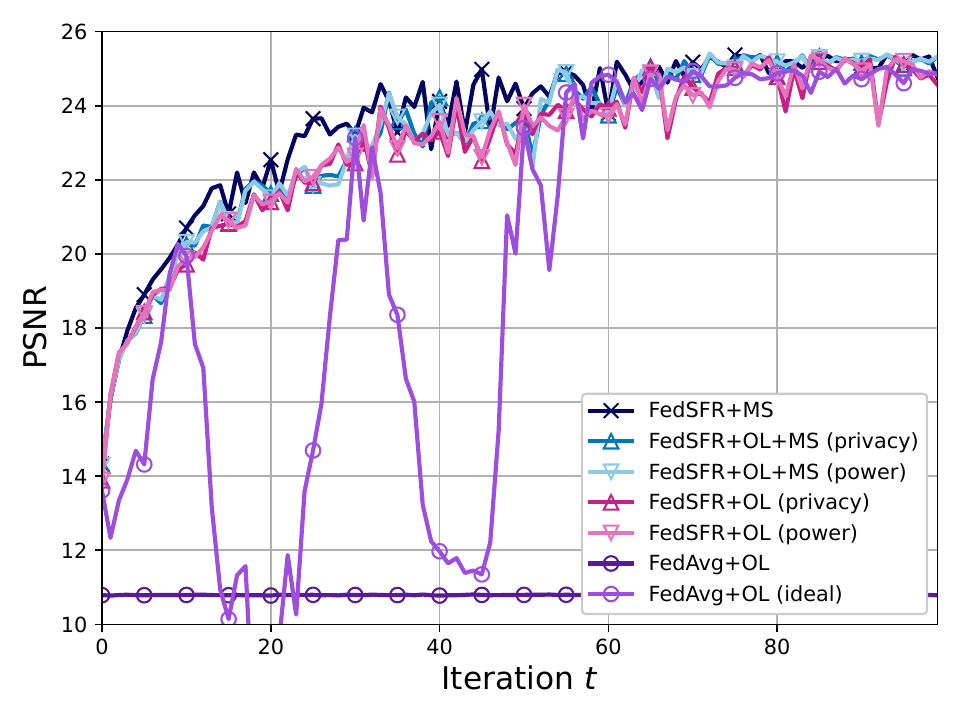}
    \label{fig:cifar dp ll}}
    \subfigure[$(\sigma_1, \sigma_2) = (0.001, 0.001)$]{\includegraphics[width=0.3\linewidth]{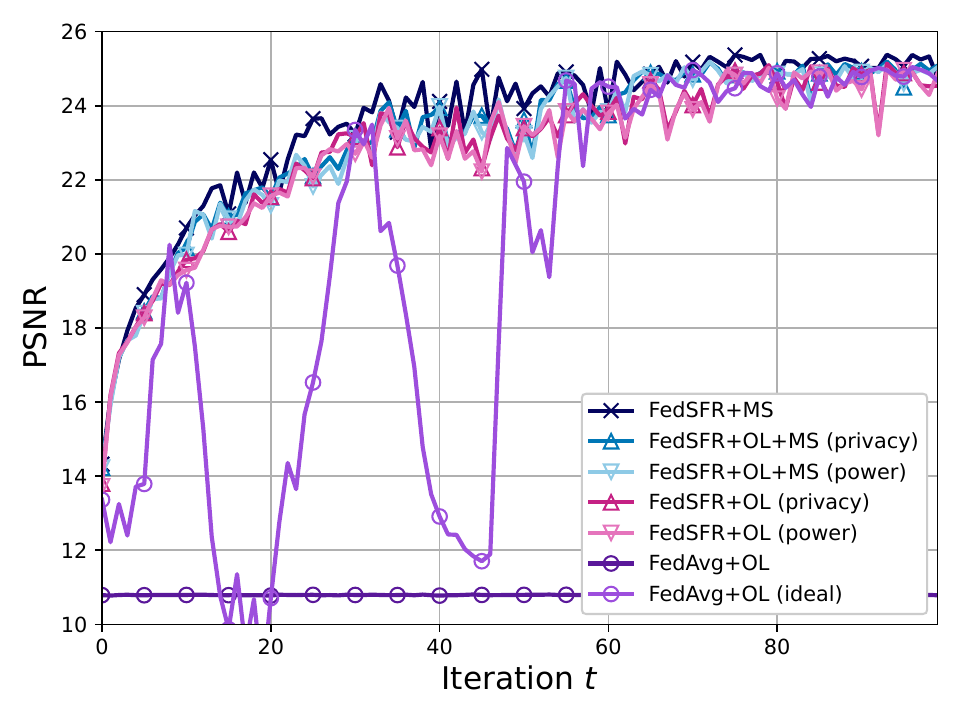}
    \label{fig:cifar dp mm}}
    \subfigure[$(\sigma_1, \sigma_2) = (0.01, 0.01)$]{\includegraphics[width=0.3\linewidth]{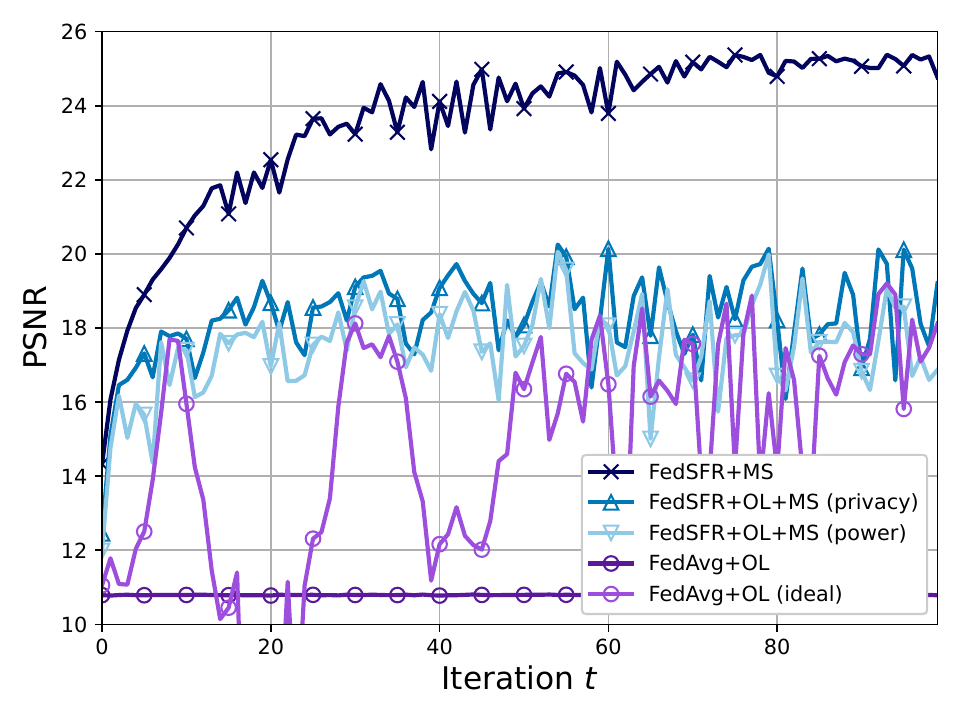}
    \label{fig:cifar dp hh}}
    \vspace{-3mm}
    \caption{Performance comparison between differentially private FedSFR and FedAvg with $\sigma_1 = \sigma_2$ for CIFAR-10 dataset.}
    \label{fig:cifar dp 1}
    \vspace{-5mm}
\end{figure*}

\begin{figure*}[!t]
    \centering
    \setlength{\subfigcapskip}{-2mm}
    \subfigure[$(\sigma_1, \sigma_2) = (0.0001, 0.01)$]{\includegraphics[width=0.3\linewidth]{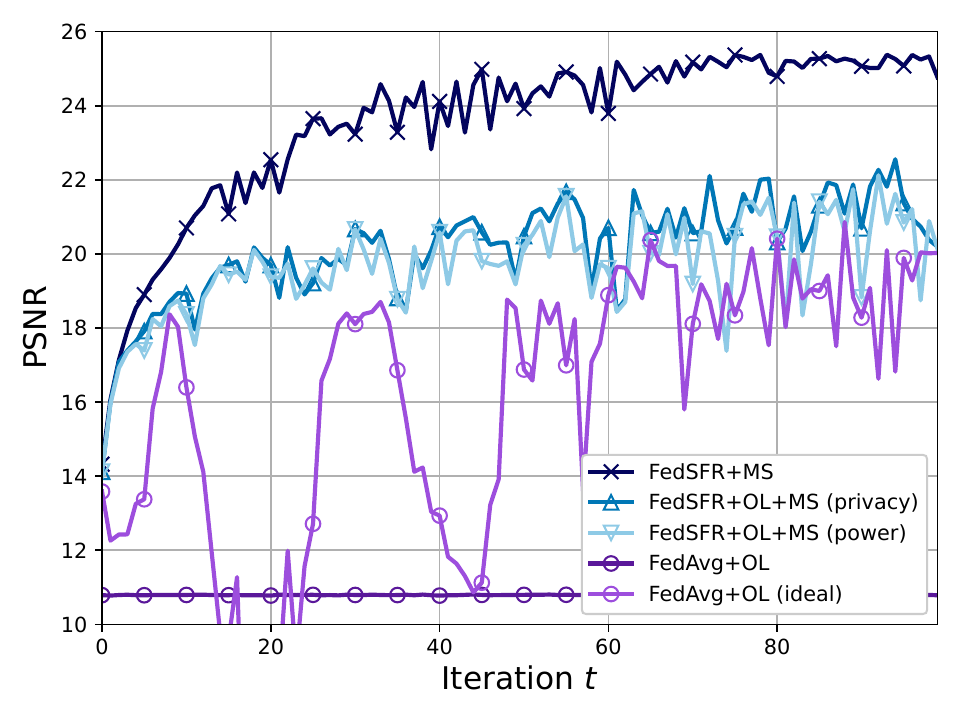}
    \label{fig:cifar dp lh}}
    \subfigure[$(\sigma_1, \sigma_2) = (0.02, 0.00005)$]{\includegraphics[width=0.3\linewidth]{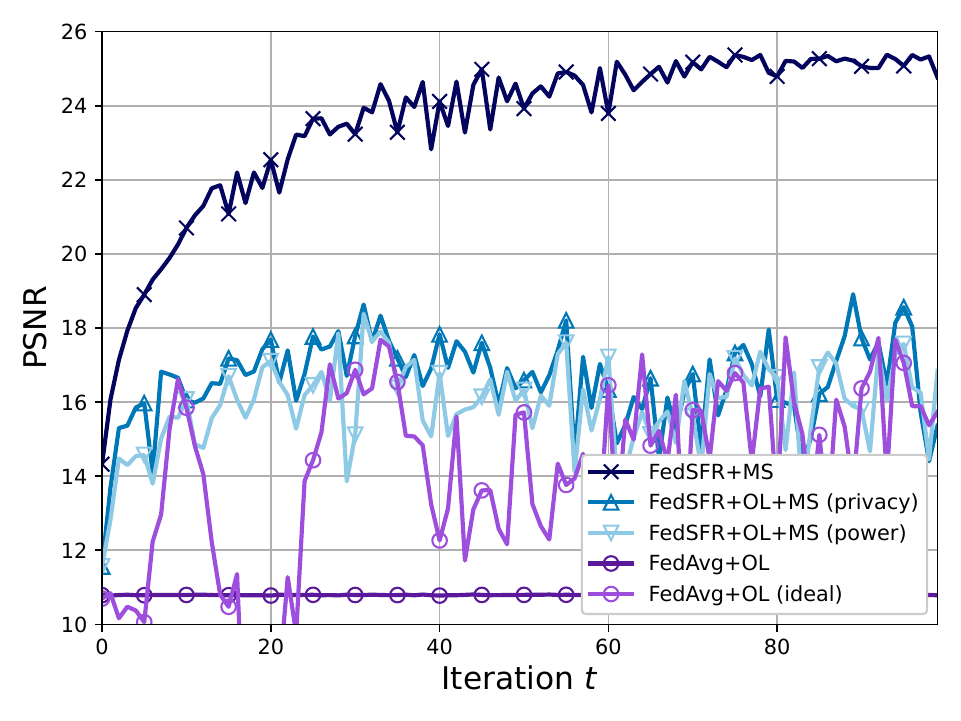}
    \label{fig:cifar dp hl}}
    \subfigure[SNR vs. PSNR]{\includegraphics[width=0.3\linewidth]{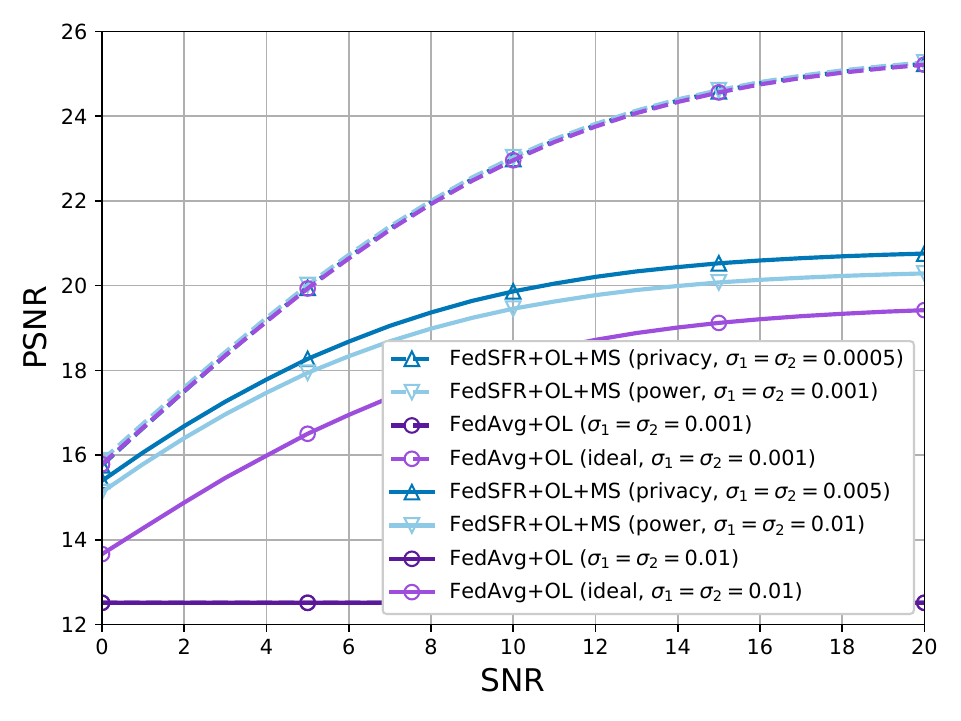}
    \label{fig:cifar dp psnr}}
    \vspace{-3mm}
    \caption{Performance comparison between differentially private FedSFR and FedAvg (a), (b) with $\sigma_1 \neq \sigma_2$, and (c) under Rayleigh fading channels with varying average SNR for CIFAR-10 dataset.}
    \label{fig:cifar dp 2}
    \vspace{-6mm}
\end{figure*}

We compare the proposed differentially private FedSFR framework with the FedAvg baseline using the CIFAR-10 dataset, which consists of $3\times32\times32$ RGB images. Reconstruction performance is evaluated in terms of the peak signal-to-noise ratio (PSNR). $K_m + K_o = 20$ clients are chosen from the total $K = 50$ clients, following the common partial participation setting in FL to reduce uplink burden \cite{huh2025feature, huh2026federated, kim2025privacy}. The sparsification level $S_k$ varies at each iteration in proportion to the channel capacity between the PS and client $k$ under Rayleigh fading with $\mathcal{CN}(0, 1)$, yielding $S_k/N$ approximately between $0.05$ and $0.25$. The threshold determining whether client $k$ is included in $\mathcal{A}_m^{\scriptstyle(t)}$ or $\mathcal{A}_o^{\scriptstyle(t)}$ is set to $S_\mathsf{th}/N = 0.15$, resulting in $K_o \approx 10$. The initial learning rates are $\eta_s^{\scriptstyle(0)} = 0.001$ at the PS and $\eta_c^{\scriptstyle(0)} = 0.01$ at the clients, both decayed by a factor of $0.9$ every $10$ global iterations. Both clients and the PS use batch size $16$ for $3$ epochs. The JSCC encoder and decoder are implemented with five convolutional and five transposed convolutional layers, respectively, resulting in about $0.32$M trainable parameters. The feature dimension is set to $d = 256$. The training SNR $\gamma$ is fixed at $20$ dB, and PSNR is evaluated under the same condition.

For DP evaluation, we adopt a baseline setting where Laplace noise $\mathrm{Lap}(0,\sigma_1)$ and $\mathrm{Lap}(0,\sigma_2)$ are applied to the oneshot and Laplace mechanisms with $(\sigma_1,\sigma_2) = (0.001, 0.001)$. According to \textbf{Theorem \ref{thm:differential privacy}}, \textbf{FedSFR+OL+MS (privacy)} denotes the configuration that matches the privacy level of \textbf{FedAvg+OL}, i.e., $\epsilon^{(i)} = \epsilon^{(ii)}$, by using reduced noise scales $(\sigma_1,\sigma_2) = (0.0005, 0.0005)$. In contrast, \textbf{FedSFR+OL+MS (power)} corresponds to the setting that provides twice the privacy level of \textbf{FedAvg+OL}, i.e., $\epsilon^{(i)} = 2\epsilon^{(ii)}$, by injecting noise with $(\sigma_1,\sigma_2) = (0.001, 0.001)$.

\underline{\textit{\textbf{FedSFR vs. FedAvg.}}}\; Fig. \ref{fig:cifar dp 1} compares the PSNR of DP-aided FedSFR and FedAvg with $\sigma_1 = \sigma_2$. Since \textbf{FedAvg+OL} fails to be properly trained under our heterogeneous wireless environments, we add \textbf{FedAvg+OL (ideal)}, where $S_k/N = 0.2$ for $\mathcal{A}_m^{\scriptstyle(t)}$ and $S_k/N = 0.1$ for $\mathcal{A}_o^{\scriptstyle(t)}$ with $K_o = 10$. As shown in Figs. \ref{fig:cifar dp ll} and \ref{fig:cifar dp mm}, \textbf{FedSFR+OL+MS} achieves faster convergence and more stable training than \textbf{FedAvg+OL (ideal)}, benefiting from FR, and our method is much more robust to heterogeneous settings than \textbf{FedAvg+OL}. When comparing \textbf{FedSFR+OL+MS} with \textbf{FedSFR+OL}, the proposed MS algorithm enhances both training stability and task performance due to more optimal training trajectory, particularly in the later stages of training. As illustrated in Fig. \ref{fig:cifar dp hh}, even under severe noise injection, where overall performance degradation is inevitable, the proposed method consistently outperforms the baselines. In addition, \textbf{FedSFR+OL+MS (privacy)} exhibits a more favorable training behavior than \textbf{FedSFR+OL+MS (power)}, which highlights the tradeoff between privacy guarantees and model performance.

Figs. \ref{fig:cifar dp lh} and \ref{fig:cifar dp hl} provide the same privacy guarantee by adjusting $\sigma_1$ and $\sigma_2$. A comparison between the two figures shows that increasing $\sigma_1$ while decreasing $\sigma_2$ degrades the performance of both frameworks. This trend indicates that heavy noise injected during the top-$S$ index selection ($\sigma_1$) disrupts the identification of important local updates more severely than magnitude perturbations ($\sigma_2$). Nevertheless, even under this unfavorable noise distribution, FedSFR consistently outperforms FedAvg, highlighting the robustness of FR. Beyond the effect of the MS algorithm, these observations indicate that, given privacy requirement, proper control of $\sigma_1$ and $\sigma_2$ plays a critical role in performance and warrants further investigation. To compare the semantic communication modules from FL, Fig. \ref{fig:cifar dp psnr} illustrates the PSNR under Rayleigh fading channels. When $\sigma_1 = \sigma_2 = 0.001$ (dashed lines), all methods exhibit similar PSNR due to their comparable final FL performance except \textbf{FedAvg+OL}, following Fig. \ref{fig:cifar dp mm}. However, when $\sigma_1 = \sigma_2 = 0.01$ (solid lines), FedSFR achieves better performance than FedAvg, and \textbf{FedSFR+OL+MS (privacy)} outperforms \textbf{FedSFR+OL+MS (power)}, following Fig. \ref{fig:cifar dp hh}.

\begin{figure*}[!t]
    \centering
    \setlength{\subfigcapskip}{-2mm}
    \subfigure[CelebA dataset]{\includegraphics[width=0.3\linewidth]{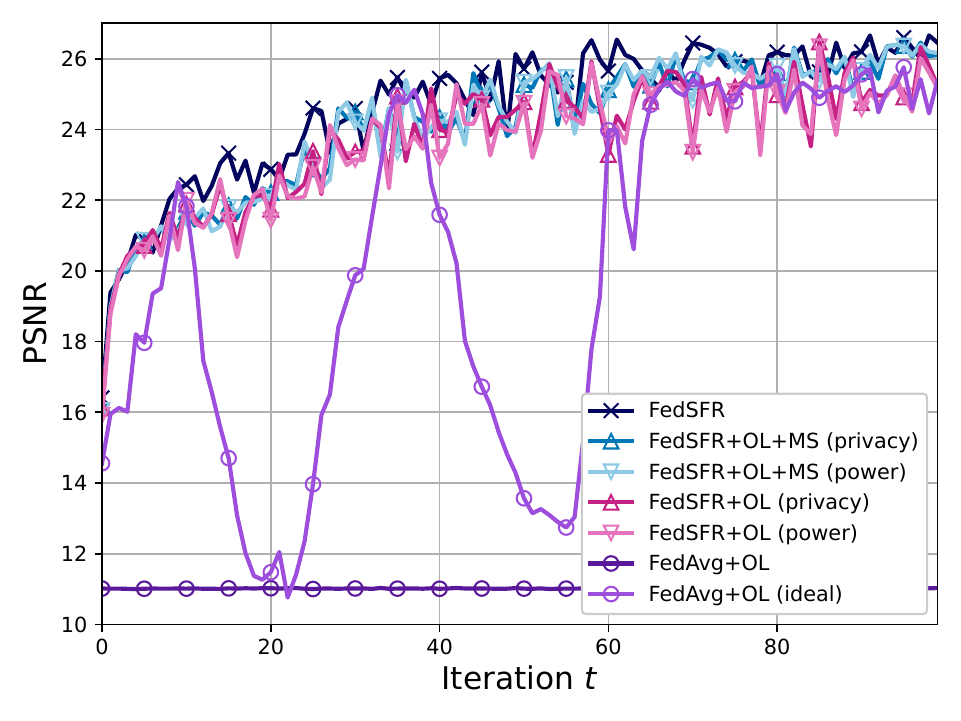}
    \label{fig:celeb dp mm}}
    \subfigure[$10$dB training SNR]{\includegraphics[width=0.3\linewidth]{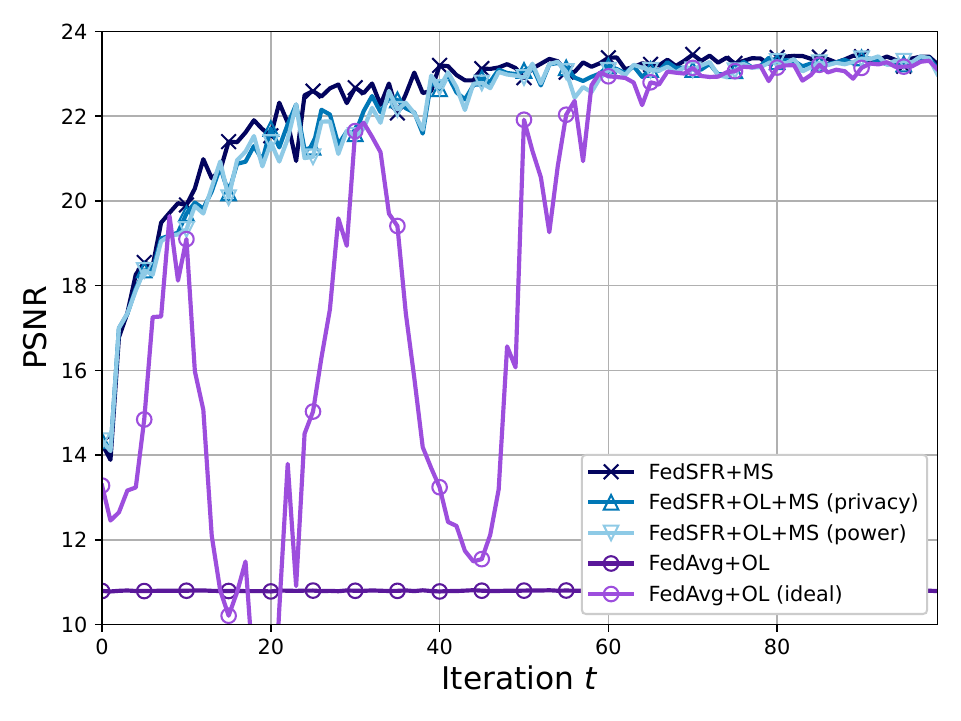}
    \label{fig:cifar dp 10}}
    \subfigure[Varying $S_\mathsf{th}$]{\includegraphics[width=0.3\linewidth]{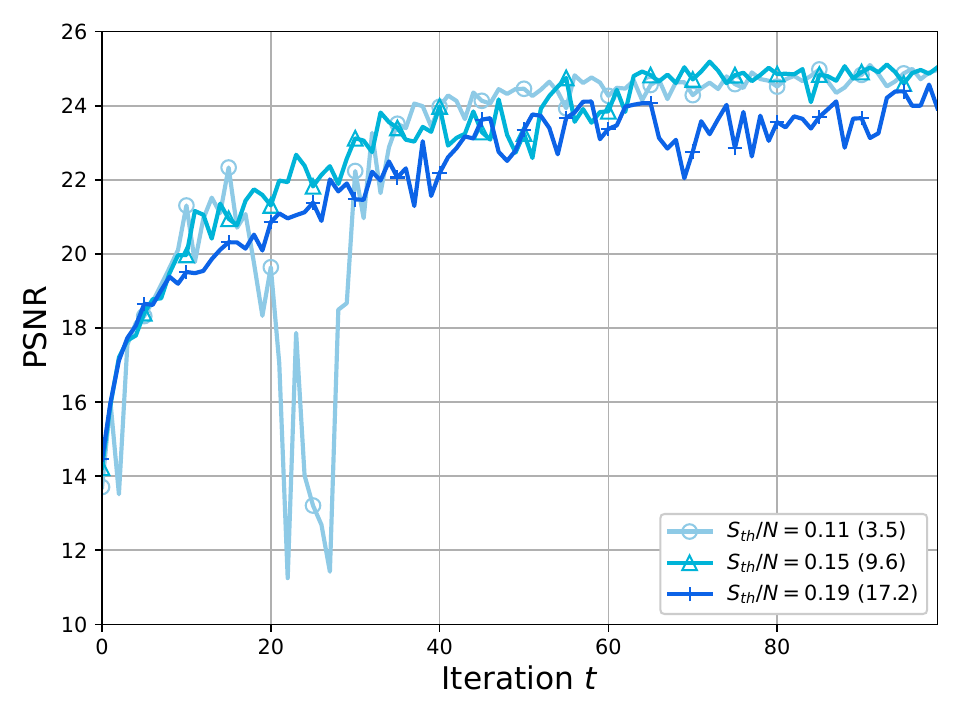}
    \label{fig:cifar dp rayleigh}}
    \vspace{-3mm}
    \caption{Performance comparison between differentially private FedSFR and FedAvg (a) for CelebA dataset, (b) under $10$dB training SNR, and (c) performance of \textbf{FedSFR+OL+MS (power)} with varying the threshold $S_\mathsf{th}$ for CIFAR-10 dataset.}
    \label{fig:dp ablation}
    \vspace{-6mm}
\end{figure*}

\underline{\textit{\textbf{Ablation Studies.}}}\; To assess the applicability of the proposed approach to higher-resolution settings, we perform experiments on the CelebA dataset, as illustrated in Fig. \ref{fig:celeb dp mm}. The experimental outcomes follow the same behavior as those obtained on the CIFAR-10 dataset, as shown in Fig. \ref{fig:cifar dp mm}, confirming the consistency of the proposed framework. Moreover, similar to the $20$dB setting in Fig. \ref{fig:cifar dp mm}, the $10$dB case in Fig. \ref{fig:cifar dp 10} surpasses the baselines, while converging faster due to its lower achievable performance. Fig. \ref{fig:cifar dp rayleigh} illustrates the impact of the threshold $S_\mathsf{th}$ and $|\mathcal{A}_o^{\scriptstyle(t)}|$ averaged over $T$ iterations is indicated in the legend. When $S_\mathsf{th}/N = 0.11$, the reduced FR effects due to fewer clients in $\mathcal{A}_o^{\scriptstyle(t)}$ result in unstable initial training. For $S_\mathsf{th}/N = 0.19$, the PSNR becomes the lowest because excessive FR occurs with fewer clients in $\mathcal{A}_m^{\scriptstyle(t)}$. The setting $S_\mathsf{th}/N = 0.15$ achieves the highest PSNR with a stable training process, yielding balanced FL behavior. These suggest that the proper choice of the threshold plays an important role in performance and deserves further investigation.

\section{Conclusion}\label{sec:conclusion}
This paper proposed a DP-aided FedSFR framework, which incorporates the OL mechanism, for updating semantic communication modules in heterogeneous wireless systems. By leveraging the inherent privacy benefit of clients that transmit locally generated semantic feature vectors instead of sparse local updates, we proved that feature-based transmission achieves strictly stronger DP guarantees than gradient-based transmission under an identical uplink budget. To mitigate the performance degradation induced by DP noise, we proposed the MS algorithm that adaptively retains beneficial server-side updates. Experiments demonstrated that the proposed DP-aided FedSFR consistently outperforms DP-enabled FedAvg in terms of image reconstruction performance and training stability, while maintaining rigorous privacy guarantees.

\bibliographystyle{IEEEtran}  
\bibliography{IEEEabrv,reference}

\end{document}